\documentclass[runningheads]{llncs}
\usepackage[english]{babel}
\usepackage[T1]{fontenc}
\usepackage{amsmath,amsfonts,amssymb}
\usepackage[a4paper,top=2cm,bottom=2cm,left=3cm,right=3cm,
            marginparwidth=1.75cm]{geometry}
\usepackage{graphicx,float}
\usepackage[colorlinks=true,allcolors=blue]{hyperref}
\usepackage{etoolbox}
\usepackage{orcidlink}

\newcommand*{\Tr}{\operatorname{Tr}}
\newcommand*{\Z}{\mathbb{Z}}
\newcommand*{\F}{\mathbb{F}}
\newcommand*{\ev}{\operatorname{ev}}

\begin{document}

\title{Extensions of
root-based attacks against PLWE via isomorphisms of rings}
\authorrunning{I. Blanco Chac\'on, R. Dur\'an D\'{\i}az, R. Mart\'{\i}n S\'anchez-Ledesma}
\author{Iv\'an Blanco Chac\'on\orcidlink{0000-0002-4666-019X}\inst{1}
       \and
        Ra\'ul Dur\'an D\'{\i}az\orcidlink{0000-0001-6217-4768}\inst{2}
       \and
       Rodrigo Mart\'{\i}n S\'anchez-Ledesma\orcidlink{0009-0001-1845-2959}\inst{3,4}}

\institute{Departamento de F\'{\i}sica y Matem\'aticas,
           Universidad de Alcal\'a, Spain \\
           \email{ivan.blancoc@uah.es}
           \and
           Departamento de Autom\'atica,
           Universidad de Alcal\'a, Spain \\
           \email{raul.duran@uah.es}
           \and
           Departamento de \'Algebra,
           Universidad Complutense de Madrid, Spain \\
           \email{rodrma01@ucm.es}
           \and
           Indra Sistemas de Comunicaciones Seguras, Spain \\
           \email{rmsanchezledesma@indra.es}}

\maketitle
\def\figurename{Algorithm}

\begin{abstract}
In the present work we address some key questions regarding the generalization
of root-based attacks presented in \cite{BDM:2025:GARB}. In particular, we
analyze potential root-based attacks extensions via the construction of
explicit isomorphisms from vulnerable instances, and provide a formal proof that
this approach will not yield any new vulnerabilities. To do so, we first construct an explicit isomorphism between fully-split
polynomial rings and polynomial rings where previous attacks apply and
show that the application of such an isomorphism will always
distort the samples in a way that the resulting samples cannot be used to
distinguish. Then, we prove that any isomorphism between fully-split polynomial
rings must be of the form of the constructed isomorphism. Lastly, we study the generalization beyond fully split settings and show that the same distortion happens, regardless of the degree of the field extension.
\end{abstract}
\keywords{PLWE \and Root-based attacks \and Isomorphisms}

\section{Introduction}
The Ring Learning With Errors problem (RLWE) and the Polynomial Learning With
Errors problem (PLWE) have become the most important paradigms in the quest of
achieving post-quantum security. Some of the strongest theoretical clues
pointing in this direction are the \emph{worst case-average case reductions} from an
approximate version of the Shortest Vector Problem on ideal lattices to the RLWE
problem, established in \cite{LPR:2013:ILL}, and to the PLWE problem over
power-of-two cyclotomic polynomials, established in \cite{SSTX:2009:EPK}.

An important debate in cryptography in general, and post-quantum cryptography
in particular, is the one contrasting efficiency with security. Cryptographic
schemes must obviously be secure, but they need also be practical and usable.
This dichotomy is not a stranger to lattice-based cryptography,
of which both PLWE and RLWE are representatives.

Within lattices, this dilemma is represented by the definition of both
\emph{unstructured} and \emph{structured} schemes. By \emph{structured},
it is meant that additional algebraic structure is added to the scheme, whereas
\emph{unstructured} schemes normally rely on lighter mathematical
constructions.

The definition of these structured paradigms permits
cryptographic schemes employing such structures to enjoy a number of privileges,
most notably related to better performances and smaller cryptographic sizes.
But the cost of doing so is the addition of the aforementioned heavy algebraic
structure, which could potentially turn out to be the source of new attacks,
and gives rise to some important questions: \emph{are there attacks
that can effectively target this additional structure?} \emph{Do they improve
the best known attacks against LWE?}

It was precisely this line of thought that the works of
\cite{ELOS:2015:PWI,ELOS:2016:RCN} explored. In them, the authors
introduce a general framework to exploit the algebraic structure of PLWE schemes
in order to mount \emph{decisional} attacks. In particular, evaluation
homomorphisms are constructed from a chosen root $\alpha \in \mathbb{F}_q$ of
the polynomial $f(x)$ that is at the heart of the PLWE instance. Then, this map
is applied to every sample given and a number of distinguishability conditions
giving rise to different attacks do appear. Each distinguishing condition
defines a specific attack.

The work of \cite{Peikert:2016:HIR} studied the likelihood of some of the
attacks presented in \cite{ELOS:2015:PWI}. A statistical analysis showed that
the distribution of the setting potentially vulnerable to the attacks
under consideration represented a negligible share of the total instances, thus seemingly
closing the door to this approach. However, it is worth noting that this analysis did not
include the attack referred to as \emph{Unbounded Small Values Attack}
\cite{BDM:2025:GARB}, which newer analysis have found to be the most likely to
be applicable. The latter work further included new techniques for decisional attacks
against the algebraic structure.

In particular, \cite{BDM:2025:GARB} provides two advances: first,
a deeper analysis of the previously defined attacks by both refining the given
success probabilities and introducing another analysis in terms of a posteriori
probabilities. The refined bounds on the success probabilities of the attacks
require introducing new techniques on them, to make them useful.

Second, their work generalized the attacks so that they become applicable to
polynomials whose factorization is defined only over extensions of $\mathbb{F}_q$.
In other words, the proposed attacks work for any root belonging to a finite
extension $E$ of $\mathbb{F}_q$ of degree $k$, for any $k \geq 1$.

The key ingredient is the trace operator defined over a field extension, in
order to complement the evaluation homomorphism. The latter allows the
transfer of the setting to $\mathbb{F}_q$ but remark that so does the former:
actually the trace operator is affine over any extension $E$ of $\mathbb{F}_q$,
\emph{i.e.}, it respects the sum of elements in $E$ and the product by elements of the
base field $\mathbb{F}_q$.

The definition of higher degree extension is done under a special kind of
factorization, named \emph{$k$-ideal factors} (see Definition~\ref{d.3}). They
represent polynomials
that have factors of the form $x^k - a$, where $a \in \mathbb{F}_q$. While the
limitation to roots of ideal factors is indeed an applicability restriction, it
is important to note that, in practice, most cryptographically relevant PLWE
schemes follow this type of factorization. This is due to the fact that
ideal factors are preferred for efficiency purposes, specially related to
polynomial multiplications (consider, for example, the well-known \emph{Number
Theoretic Transform} whose applicability relies precisely on the existence of
such ideal factors).

This framework provides an initial guess to the questions previously asked:
\emph{Are there attacks effectively targeting this additional
structure?} The answer is \emph{Yes}, since these attacks are built precisely
from the algebraic structure of the PLWE paradigm, and cannot be migrated to
purely LWE settings.

The next pending question was
\emph{Do they improve the best known attacks against LWE?}
Really, most
of the attacks defined poses almost no threat to cryptographically relevant
PLWE instances, either due to the likelihood of the attacks being applicable
(e.g. \emph{Bounded Small Values Attack}), or the negligible success
probability of them (e.g. \emph{Small Errors Attack}).

However, in \cite{ABBHHS:2024:FMP,BHM:2026:FMA} comparative analysis between maximal real and
cyclotomic polynomials showed that, while still negligible, the
\emph{Unbounded Small Value Attack} has a much more meaningful impact radius,
when compared with the other attacks. More so, this attack was the only one to
be prone to positive attack results.

These findings could give an initial suggestion that there is not a meaningful
advantage in trying to target the additional algebraic structure, at least
through this approach.

There are other more recent attacks targeting the RLWE that circumvent the
RLWE/PLWE- equivalence. The most prominent one is \cite{CDW:2017:SSC}, which
drastically reduces the $\gamma$-approximation factor for the underlying SVP
keeping polynomial quantum complexity. This attack only applies to the ring of
integers of cyclotomic number fields (or sub-lattices therein) and uses a number
of relations of the so called Stickelberger ideal. It is a much interesting
question whether this attack can be applied to arbitrary Abelian sub-extensions
of cyclotomic rings.

\subsection{Our contributions}
The present work, pursuing the same research line, intends to answer two
related questions that remain still open.

The first one concerns the applicability of these attacks to polynomials with
factorizations beyond plain $k$-ideal factors.
In Section \ref{S.3}, it will be shown that, in theory, the attacks defined in
\cite{BDM:2025:GARB} can indeed be transferred to this more general
framework. However, in practice, this generalization to arbitrary
factorization impacts negatively the applicability of the attacks and the
resulting success probabilities.

The second and more important question, comes from the idea of migrating to
settings that might have more favorable conditions. Formally, the notion
behind this idea is to employ isomorphisms of finite fields to migrate
from the target PLWE to a vulnerable setting, apply a distinguishability
attack there, and infer the distribution of the original PLWE scheme from that
attack.

The drawback that one could think of when considering this idea is the fact
that such transference of samples from one setting to another might potentially
increase the noise of the samples so as to spoil completely the
distinguishability in the favorable PLWE setting.

This work investigates a natural way to perform such an attach and thus backs up
this intuition: Section~\ref{S.4} will show precisely that no isomorphism
between PLWE instances provides any meaningful advantage over attacking the
original PLWE instance.

In order to pursue this idea, a specific isomorphism between PLWE
distributions will be constructed and it will be proven that the generated
noise increase is given by the powers of the roots of the original PLWE setting.
Moreover, it will be proven that this specific type of isomorphism is the only
possible one between PLWE settings. While this analysis is only done for PLWE
polynomials which are fully split (\emph{i.e.} they factor completely in
$\mathbb{F}_q$), this remains the case for most cryptographically relevant PLWE
instances. These results will be presented in Appendix~\ref{a.1}. 

Then, Section \ref{S.4} works up that intuition for settings beyond fully-split ones, proving that this approach is not applicable, regardless of the degree of the extension of finite fields defined by the PLWE instance.

In conclusion, our work clearly backs up the intuition that if any root-based
attack against
the initial PLWE instance is unsuccessful, so will also be any other attack
even when the samples are transferred to a new (potentially weaker) PLWE setting:
actually, the noise increase will be high enough so as to render useless attacks
otherwise known to be successful under such setting.

\section{Preliminaries}\label{S.2}
This section provides an overview of the PLWE problem, and the essentials
about \emph{root-based attacks} needed to comprehend the following sections.

We begin with the definition of the PLWE distribution.
\begin{definition}[PLWE distribution]
Let $q$ be a rational prime, $f(x) \in \Z[x]$ a monic irreducible polynomial
and $\mathcal{O}_f$ the associated quotient ring $\Z[x]/(f(x))$. Let
$\chi$ be a discrete random distribution with values in
$\mathcal{O}_f/q\mathcal{O}_f$. For $s \in$
$\mathcal{O}_f/q\mathcal{O}_f$, we define the \emph{PLWE distribution}
\emph{$\mathcal{B}_{s, \chi}$} as the distribution over
$\mathcal{O}_f/q\mathcal{O}_f \times \mathcal{O}_f/q\mathcal{O}_f$ obtained by
sampling an element $a$ in $\mathcal{O}_f/q\mathcal{O}_f$ uniformly at random,
sampling an element $e$ from $\chi$, and returning the pair $(a, a\cdot
s + e)\in(\mathcal{O}_f/q\mathcal{O}_f\times\mathcal{O}_f/q\mathcal{O}_f)$.
\end{definition}

Next the definition of what will be referred to as a $n$-ideal factor:
\begin{definition}\label{d.3}
Let $q$ be a rational prime. Let $f(x) \in \Z[x]$ be a monic irreducible
polynomial of degree $N$. Suppose that $a \in \F_q$ has small
multiplicative order in $\F_q^*$ and that there exists $1 \leq n < N$ such that
the (irreducible) polynomial $x^n-a \in \F_q[x]$ divides $f(x) \bmod{q}$. We call $x^n-a$ a
$n$-\emph{ideal factor} of $f(x) \bmod{q}$.
\end{definition}

The following setting will be assumed for the introduction of \emph{root-based
attacks:} Let $q$ be a prime and let $f(x)\in \mathbb{Z}[x]$ a polynomial
over $\mathbb{Z}[x]$ of degree $N$ satisfying Definition~\ref{d.3} above.
Denote $R_q$ the ring
\[
R_q := \mathbb{F}_q[x]/(f(x))
\]
and $R_{q,0}\subseteq R_q$ the sub-ring
\[
R_{q,0} := \{p(x)\in R_q: p(\alpha)\in\mathbb{F}_q\},
\]
where $\alpha$ is a root of $f(x)$ in a certain field extension, to be defined
in the next paragraphs.

For the PLWE distribution, we assume a Gaussian distribution of mean $0$ and a
certain variance $\sigma^2$. Define $p_0$ as the probability of any random
Gaussian variable with the above distribution to lie inside $[-2\sigma,
2\sigma]$. The value $p_0$ is known to be $0.954499$, up to 5 digits of
precision.

The general framework of root-based distinguishability attacks of
\cite{BDM:2025:GARB} can be summarized as follows:
\begin{enumerate}
    \item Let $f(x)$ be a polynomial satisfying the conditions in
Definition~\ref{d.3} and let $\alpha$ be a root of such polynomial $f(x)$.
    \item Assume we have a linear map
\[
\Phi_\alpha\colon R_q\rightarrow\mathbb{F}_q.
\]
    \item Given samples $(a_i(x), b_i(x) = s(x)\cdot a_i(x) + e_i(x))$,
consider them as elements of $R_q$. Since $\Phi_\alpha$ is linear, the
following equality holds
\[
\Phi_\alpha(e_i(x)) = \Phi_\alpha(b_i(x) - s(x) \cdot a_i(x)) =
                      \Phi_\alpha(b_i(x)) - \Phi_\alpha(s(x) \cdot a_i(x)).
\]

    \item Accordingly, fixing a guess for $s(x)$, compute the value
$\Phi_{\alpha}(e_i(x))$ associated to each sample $(a_i(x),b_i(x))$.
    \item Perform certain distinguishability actions over the tentative
evaluated errors to get a distinguishability feature.
\end{enumerate}

The established conditions regarding $f(x)$ ensure that $x^n - a$ defines a
field extension $E$ of degree $n \geq 1$ over $\mathbb{F}_q$. This extension
defines the \emph{trace operator} $\Tr_{E | \mathbb{F}_q}: E \rightarrow
\mathbb{F}_q$ of $E$. Moreover, any root $\alpha$ of $x^n - a$ defines an
\emph{evaluation ring homomorphism} $\phi_{\alpha}: R_q \rightarrow E$.
Composing the maps described, we define $\Phi_{\alpha}$ as
\[
\Phi_{\alpha} := \frac{1}{n}\Tr \circ \phi_{\alpha},
\]
which is certainly linear since the trace operator is additive.
Note that even in the case $n = 1$, the map is still well defined, as the trace
operator becomes the identity over $\mathbb{F}_q$.

The particular feature used to distinguish the tentative evaluated errors will
define each of the attacks. Three different attacks will be specified:
\begin{enumerate}
    \item \emph{Small Set Attack}: Analyze the set of values resulting from
the application of the map $\Phi_{\alpha}$ to the set of all possible error
values when this set has small cardinality.
    \item \emph{Bounded Small Values Attack}: Analyze the smallness of the set
of values resulting from the application of the map $\Phi_{\alpha}$ to a
distribution of error values constrained to a certain interval.
    \item \emph{Unbounded Small Values Attack}: Analyze the smallness of the set
of values resulting from the application of the map $\Phi_{\alpha}$ to a
distribution of error values not constrained to any particular interval.
\end{enumerate}

\subsection{Small Set Attack}
The \emph{Small Set Attack} sub process is defined as follows:
\begin{figure}[H]
\centering
\begin{tabular}[c]{ll}
\hline
\textbf{Input:} & A set of samples $C=\{(a_i(x),b_i(x))\}_{i=1}^M\in
R_{q,0}\times R_q$ \\
 & A look-up table $\Sigma$ of values appearing in $\frac{1}{n}
\Tr(b(\alpha) - a(\alpha)\cdot s)$ \\
 & with probability $\geq p_0^r$ \\
\textbf{Output:} & \textbf{PLWE},\\
                 & or \textbf{NOT PLWE},\\
                 & or \textbf{NOT ENOUGH SAMPLES}\\
\hline
\end{tabular}

\medskip

\begin{itemize}
\item $G:=\emptyset$
\item \texttt{\emph{for}} $g\in \mathbb{F}_q$ \texttt{\emph{do}}
	\begin{itemize}
	\item \texttt{\emph{for}} $(a_i(x),b_i(x))\in C$  \texttt{\emph{do}}
		\begin{itemize}
		\item \texttt{\emph{if}} $\frac{1}{n}\left(\Tr(b_i(\alpha))
                                -a_i(\alpha)g\right)\notin\Sigma$
             \texttt{\emph{then next}} $g$
		\end{itemize}
	\item $G:=G\cup\{g\}$
	\end{itemize}
\item \texttt{\emph{if}} $G=\emptyset$ \texttt{\emph{then return}} \textbf{NOT
PLWE}
\item \texttt{\emph{if}} $|G|=1$ \texttt{\emph{then return}} \textbf{PLWE}
\item \texttt{\emph{if}} $|G|>1$ \texttt{\emph{then return}} \textbf{NOT ENOUGH
SAMPLES}
\end{itemize}
\hrule
\caption{Attack based on the size of the set of possibilities for the evaluated errors}
\label{SSAAlg}
\end{figure}
However, due to the limit in the success probabilities of this algorithm, given
in \cite[\S 4.2, \S 4.3]{BDM:2025:GARB}, the following algorithm was defined, which we will
refer to as \emph{Extended Small Set Attack}:
\begin{figure}[H]
\centering
\begin{tabular}[c]{ll}
\hline
\textbf{Input:} & A collection of samples $S=\{(a_i(x),b_i(x))\}
                                           _{i=1}^M \in R_{q, 0} \times R_q$ \\
& A choice $M_0$ for the number of samples of the sub-process \\
\textbf{Output:} & A guess for the distribution of the samples, either \\
                 & \textbf{PLWE} or \textbf{UNIFORM}\\
\hline
\end{tabular}

\medskip

\begin{itemize}
\item $T:=\lceil \left\lfloor M/M_0\right\rfloor \cdot p_0^{M_0r} \rceil$
\item $C:=0$
\item \texttt{\emph{for}} $j\in \{0, \dotsc, \left\lfloor M/M_0\right\rfloor - 1\}$ \texttt{\emph{do}}
	\begin{itemize}
	\item $result := \mathrm{SmallSetAttack}\left(S_{j} := \{(a_i(x),b_i(x))\}
                         _{i=j\cdot M_0 + 1}^{(j+1)\cdot M_0}\right)$
		\begin{itemize}
		\item \texttt{\emph{if}} $result$ $\neq$ \textbf{NOT PLWE}
            \texttt{\emph{then}}
			\begin{itemize}
			\item $C := C + 1$
			\end{itemize}
		\end{itemize}
	\end{itemize}
\item \texttt{\emph{if}} $C < T$ \texttt{\emph{then return}} \textbf{UNIFORM}
\item \texttt{\emph{else return}} \textbf{PLWE}
\end{itemize}
\hrule
\caption{Algorithm for Extended Small Set Attack}
\label{Extended_SSAAlg}
\end{figure}

\begin{proposition}\cite{BDM:2025:GARB}\label{Extended_SSA_Success}
Assume $|\Sigma|<q$, and let $M$ be the number of samples given, $M_0$ the
number of samples employ for each sub process, $T$ the threshold of the attack
and $r$ the order of the term $a$ of the ideal factor $x^n - a$. Then, we have
that
\begin{enumerate}
\item
If the samples are PLWE, Algorithm~\ref{Extended_SSAAlg} guesses correctly the
distribution of the samples with
probability at least $$1 - F(T-1, \left\lfloor M/M_0\right\rfloor, p_0^{M_0 r})$$
\item
If the samples are uniform, Algorithm~\ref{Extended_SSAAlg} guesses correctly
the distribution of the samples with
probability at least $$F(T-1, \left\lfloor M/M_0\right\rfloor,
1 - \left(|\Sigma|/q\right)^{M_0})$$
\end{enumerate}
where $F$ is defined as the \emph{Cumulative Binomial Function}
\end{proposition}

\subsection{Bounded Small Values Attack}
The \emph{Bounded Small Values Attack} sub process is defined in Algorithm
\ref{SVAAlg}:
\begin{figure}[H]
\centering
\begin{tabular}[c]{ll}
\hline
\textbf{Input:} & A collection of samples $C=\{(a_i(x),b_i(x))\}
_{i=1}^M\subseteq R_{q,0} \times R_q$ \\
\textbf{Output:} & A guess $g\in\mathbb{F}_q$ for $\Tr(s(\alpha))$,\\
 & or \textbf{NOT PLWE},\\
 & or \textbf{NOT ENOUGH SAMPLES}\\
\hline
\end{tabular}

\medskip

\begin{itemize}
\item $G:=\emptyset$
\item \texttt{\emph{for}} $g \in \mathbb{F}_q$ \texttt{\emph{do}}
    \begin{itemize}
        \item \texttt{\emph{for}} $(a_i(x),b_i(x))\in C$  \texttt{\emph{do}}
            \begin{itemize}
                  \item \texttt{\emph{if}} $\frac{1}{n}(\Tr(b_i(\alpha))-a_i
                                            (\alpha)g) \notin [-\frac{q}{4}, \frac{q}{4})$
                   \texttt{\emph{then next}} $g$
            \end{itemize}
        \item $G:=G\cup\{g\}$
    \end{itemize}
\item \texttt{\emph{if}} $G=\emptyset$ \texttt{\emph{then return}} \textbf{NOT
PLWE}
\item \texttt{\emph{if}} $G=\{g\}$ \texttt{\emph{then return}} $g$
\item \texttt{\emph{if}} $|G|>1$ \texttt{\emph{then return}} \textbf{NOT ENOUGH
SAMPLES}
\end{itemize}
\hrule
\caption{Attack based on the size of error values}
\label{SVAAlg}
\end{figure}
However, due to the limit in the success probabilities of this algorithm, given
in \cite[\S 5.2, \S 5.3]{BDM:2025:GARB}, the following algorithm was defined, which we will
refer to as \emph{Extended Small Values Attack}:
\begin{figure}[H]
\centering
\begin{tabular}[c]{ll}
\hline
\textbf{Input:} & A collection of samples $S=\{(a_i(x),b_i(x))\}
                                           _{i=1}^M \in R_{q, 0} \times R_q$ \\
& A choice $M_0$ for the number of samples of the sub-process \\
\textbf{Output:} & A guess for the distribution of the samples, either \\
                 & \textbf{PLWE} or \textbf{UNIFORM}\\
\hline
\end{tabular}

\medskip

\begin{itemize}
\item $T:=\lceil \left\lfloor M/M_0\right\rfloor \cdot p_0^{M_0} \rceil$
\item $C:=0$
\item \texttt{\emph{for}} $j\in \{0, \dotsc, \left\lfloor M/M_0\right\rfloor - 1\}$ \texttt{\emph{do}}
	\begin{itemize}
	\item $result := \mathrm{SmallValueAttack}\left(S_{j} := \{(a_i(x),b_i(x))\}
                         _{i=j\cdot M_0 + 1}^{(j+1)\cdot M_0}\right)$
		\begin{itemize}
		\item \texttt{\emph{if}} $result$ $\neq$ \textbf{NOT PLWE}
            \texttt{\emph{then}}
			\begin{itemize}
			\item $C := C + 1$
			\end{itemize}
		\end{itemize}
	\end{itemize}
\item \texttt{\emph{if}} $C < T$ \texttt{\emph{then return}} \textbf{UNIFORM}
\item \texttt{\emph{else return}} \textbf{PLWE}
\end{itemize}
\hrule
\caption{Algorithm for Extended Small Value Attack}
\label{Extended_SVAAlg}
\end{figure}

\begin{proposition}\cite{BDM:2025:GARB}\label{Extended_SVA_Success}
Assume $2\overline{\sigma} < \frac{q}{4}$, let $M$ be the number of samples
given, $M_0$ the number of samples employ for each sub process, and $T$ the
threshold of the attack. Then, we have that
\begin{enumerate}
\item
If the samples are PLWE, Algorithm~\ref{Extended_SVAAlg} guesses correctly the
distribution of the samples with
probability at least $$1 - F(T-1, \left\lfloor M/M_0\right\rfloor, p_0^{M_0})$$
\item
If the samples are uniform, Algorithm~\ref{Extended_SVAAlg} guesses correctly
the distribution of the samples with
probability at least $$F(T-1, \left\lfloor M/M_0\right\rfloor,
1-(\frac{1}{2}\pm\frac{1}{2q})^{M_0})$$
\end{enumerate}
where $F$ is defined as the \emph{Cumulative Binomial Function} and the term
$\frac{1}{2}\pm\frac{1}{2q}$ is defined as established in \cite[Lemma
1]{BDM:2025:GARB}.
\end{proposition}

\subsection{Unbounded Small Values Attack}
The Unbounded Small Values Attack is defined in Algorithm \ref{USVAAlg}:
\begin{figure}[H]
\centering
\begin{tabular}[c]{ll}
\hline
\textbf{Input:} & A collection of samples $S := \{(a_i(x),b_i(x))\}
_{i=1}^\ell \subseteq R_{q,0} \times R_q$ \\
 & according to a certain distribution. \\
 & A value $\delta > 0$. \\
\textbf{Output:} & A guess into the distribution of the samples, either \\
 & \textbf{PLWE} or \textbf{UNIFORM}.\\
\hline
\end{tabular}

\medskip

\begin{itemize}
\item $T:=\left\lceil\frac{1}{2}\left(\ell\cdot q + 2\ell \cdot \delta \pm \ell
          \cdot \left(1 - \frac{1}{q}\right)\right) \right\rceil$
\item $C:=0$
\item \texttt{\emph{for}} $g\in \mathbb{F}_q$ \texttt{\emph{do}}
	\begin{itemize}
	\item \texttt{\emph{for}} $(a_i(x),b_i(x))\in S$  \texttt{\emph{do}}
		\begin{itemize}
		\item \texttt{\emph{if}} $\frac{1}{n}(\Tr(b_i(\alpha))-a_i(\alpha)g) \in
[-\frac{q}{4}, \frac{q}{4})$
\texttt{\emph{then}}
			\begin{itemize}
			\item $C = C + 1$
			\end{itemize}
		\end{itemize}
	\end{itemize}
\item \texttt{\emph{if}} $C < T$ \texttt{\emph{then return}} \textbf{UNIFORM}
\item \texttt{\emph{else return}} \textbf{PLWE}
\end{itemize}
\hrule
\caption{Unbounded Small Values Attack}
\label{USVAAlg}
\end{figure}

\begin{proposition}\cite{BDM:2025:GARB}\label{USVA_Success}
Assume $2\overline{\sigma} \geq \frac{q}{4}$, let $M$ be the number of samples
given, $M_0$ the number of samples employ for each sub process and $T$ the
threshold of the attack. Then, we have that
\begin{enumerate}
\item
If the samples are PLWE, Algorithm~\ref{USVAAlg} guesses correctly the
distribution of the samples with
probability equal to $$
\sum_{i=0}^{\ell} \left (1 - F\left(T-i-1, \ell\cdot (q-1),\frac{1}{2} \pm
\frac{1}{2q}\right) \right) \cdot \mathrm{P}\left(B\left(\ell, \frac{1}{2} +
\delta\right) = i\right)$$
\item
If the samples are uniform, Algorithm~\ref{USVAAlg} guesses correctly the
distribution of the samples with
probability equal to $$F\left(T-1, \ell\cdot q, \frac{1}{2} \pm
\frac{1}{2q}\right)$$
\end{enumerate}
where $F$ is defined as the \emph{Cumulative Binomial Function}, $B$ is the
\emph{Binomial Distribution} and the term $\frac{1}{2}\pm\frac{1}{2q}$ is
defined as established in \cite[Lemma 1]{BDM:2025:GARB}.
\end{proposition}

\section{Generalization to arbitrary factorization of the polynomial} \label{S.3}
\subsection{Overview}
The results presented in \cite{ELOS:2015:PWI}, \cite{ELOS:2016:RCN} and later
\cite{BDM:2025:GARB} have allowed to present versions of attacks against the
decision version of PLWE, in which information about a root $\alpha$ of an
$n$-degree polynomial $f(x)\in\mathbb{F}_q[x]$ was exploited, regardless of
the finite field extension of $\mathbb{F}_q$ in which the root lived.

But, as laid out in Section \ref{S.2}, all of the attacks had a common
requirement: for an $N$-degree polynomial
$f(x)$ to have a $n$-ideal factor, \emph{i.e.}, for $f(x)$ to have $x^n - a$,
$a\in \mathbb{F}_q$, as a factor in $\mathbb{F}_q[x]$.

In the sequel we present two attempts to completely generalize the root-based
attack setting, with the aim of applying it (if proved successful) to any
arbitrary polynomial factorization. This approach could render this attack
applicable to any and every single PLWE instance, thus representing a practical
constraint when considering the deployment of PLWE constructions.

In other words, this generalization would force every parametrization of PLWE
to consider these attacks, and check that for no root of $f(x)$, any of the
attacks laid out in this work is successful.

In particular, we will study two approaches, with the following contributing
results:
\begin{enumerate}
\item
First, a note on how to further expand the attacks presented in
\cite{BDM:2025:GARB}, via lifting the restriction of working only with
$n$-ideal factors, thus allowing any suitable factor to be considered for the
distinguisher.
\item
More importantly, an analysis of the isomorphism approach, which
intuitively hopes to transform samples from a safe setting onto a (hopefully
more) favorable
one in order to mount the distinguisher attacks over the latter. This approach is of
great significance as, if achievable, it would create a new source of attacks
to settings considered otherwise safe so far.

(Un)fortunately, we will prove that any attempt at migrating the setting will
yield unsuccessful results, providing an additional source of confidence
to, at least, \emph{fully-split} PLWE instances.
\end{enumerate}

We deal now both cases in turn.

\subsection{Consideration of arbitrary factorization}
In previous works, such as \cite{BDM:2025:GARB}, it is assumed that every
polynomial $f(x)$ has an $n$-ideal factor, with $0 < n < N$.

In the finite field extension case this factor was actually not necessary for
the attacks to succeed (except for $n = 1$, where it is indeed needed in order
to ensure the selected root belonged to $\mathbb{F}_q$). It was merely a device
to maximize the distinguishability features of the attack in order to increase
the chances of its success.

In the sequel, we provide a high-level overview of how the proposed attacks
would work, should the polynomial $f(x)$ have no $n$-ideal factors modulo $q$.

To begin with, consider the following
\begin{lemma}\label{l.1}
Let $g(x)=x^n-a_{n-1}x^{n-1}-\cdots - a_0$ be irreducible over $\mathbb{F}_q$,
and $\alpha\in\mathbb{F}_{q^n}$ be a root of $g$. Let $f(x)\in\mathbb{F}_q[x]$
of degree $N$ and irreducible over $\mathbb{Z}[x]$ be a polynomial such that
$g(x)$ is a factor of $f(x)$ over $\mathbb{F}_q[x]$. For each $j<N$, let
$\beta_j = \Tr(\alpha^j) \in \mathbb{F}_q$. Then
\[
\beta_j = \left\{
\begin{array}{ll}
\Tr(\alpha^j), & 0\leq j < n, \\
& \\
\sum_{k=0}^{n-1} a_k\beta_{j-n+k}, & n\leq j < N.
\end{array}
\right.
\]
\end{lemma}
\begin{proof}
The recurrence formula is extracted directly from $g(x)$, under multiplication
by powers of $\alpha$.\qed
\end{proof}

The PLWE problem is defined over the ring $R_q:=\mathbb{F}_q[x]/(f(x))$. Given
an error polynomial $e(x) = \sum_{i=0}^{N-1}e_ix^i$ in $R_q$, we can compute
the trace of such a polynomial evaluated at the root $\alpha$ with the help of
Lemma~\ref{l.1} as follows
\[
\Tr(e(\alpha)) = \sum_{i=0}^{N-1} e_i\Tr(\alpha^i) = \sum_{i=0}^{N-1}
e_i\beta_i.
\]

Now, the PLWE problem demands that each coefficient $e_i$ be sampled from a
discretized Gaussian random variable reduced modulo $q$ of mean $0$ and variance
$\sigma^2$. Observe that,
once $\alpha$ is fixed, the traces of its powers, which live in
$\mathbb{F}_q$, are fixed as well, so that $\Tr(e(\alpha))$ can be seen as a
Gaussian random variable of mean $0$ and variance
\[
\bar{\sigma}^2 = \sum_{i=0}^{N-1} \sigma^2\beta_i^2.
\]
Therefore, all the attacks presented in Section~\ref{S.2} should succeed in
the same manner, just considering this new expression for the distinguishability
features.

It is important to note that, though indeed not needing a $n$-ideal factor for
the execution of the attack,
the success probability of each of them is closely tied to the number of
$0$-valued coefficients the $n$-degree factor of $f$ has. Informally, the higher this
number, the greater the chances of a successful attack.

\section{An (unsuccessful) way of avoiding the problem: isomorphism of rings}\label{S.4}
\subsection{Reduction from arbitrary factorization to an $n$-ideal factor}
Despite the efforts to provide a totally generalized result, in which
every polynomial could have the possibility of being affected by these attacks,
it is important to remark that the chances that a polynomial $f(x)$ without an
$n$-ideal factor undergoes a successful attack are indeed slim.

This is due to
the fact that, as noted above, the number of non-zero coefficients of the
selected factor of $f(x)$ has a huge impact upon the success probability of an
attack against an $(f, q, \sigma)$-PLWE instance.

In order to increase the efficiency of the evaluation-at-$\alpha$ homomorphism
we propose yet another transformation allowing us to migrate to a more favorable
setting in which we can find vulnerable $n$-ideal factors.

To this end, we propose the concept of \emph{factorization structure}
according to the following
\begin{definition}\label{d.1}
Two $N$-degree separable polynomials, $f(x)$ and $g(x)$ are said to have the same
\emph{factorization structure} in $\mathbb{F}_q[x]$ if their factorizations,
$f(x) = \prod_{i=1}^{\ell_1} f_i(x)$, and $g(x) = \prod_{i=1}^{\ell_2} g_i(x)$,
with each $f_i(x), g_i(x)$ irreducible over $\mathbb{F}_q[x]$, are such that
$\ell_1 = \ell_2 = \ell$, and for any $i\in\{1,\dotsc,\ell\}$, there exists
$j\in\{1,\dotsc,\ell\}$ such that $f_i(x)$ and $g_j(x)$ have the same degree
(and the other way around).
\end{definition}
In other words, there exists a one-to-one correspondence between the irreducible
factors of $f(x)$ and $g(x)$ of the same degree. We can then proceed to the
following
\begin{proposition}\label{p.1}
Let $f(x), g(x)\in\mathbb{Z}[x]$ be monic polynomials of degree $N$ irreducible
in $\mathbb{Z}[x]$ having the same \emph{factorization structure} in
$\mathbb{F}_q[x]$. Then, there exists an isomorphism between the rings $R_q:=
\mathbb{F}_q[x]/(f(x))$ and $R_q^\prime := \mathbb{F}_q [x]/(g(x))$.
\end{proposition}

\begin{proof}
By the Chinese Remainder Theorem and the fact that $f_i(x), g_i(x)$ are irreducible
over $\F_q$,
\[
\begin{array}{ccccccc}
R_q & = & \mathbb{F}_q[x]/(\prod_{i=1}^\ell f_i(x)) & \simeq &
\oplus_{i=1} ^\ell\mathbb{F}_q[x]/(f_i(x)) & \simeq &
\oplus_{i=1} ^\ell\mathbb{F}_{q^{\deg(f_i(x))}} \\
& & & & \\
R_q^\prime & = & \mathbb{F}_q[x]/(\prod_{i=1}^\ell g_i(x)) & \simeq &
\oplus_{i=1} ^\ell\mathbb{F}_q[x]/(g_i(x)) & \simeq &
\oplus_{i=1}^\ell\mathbb{F}_{q^{\deg(g_i(x))}} \\
\end{array}
\]
and, since the polynomials have the same factorization structure, both will
give rise to a direct product of $\ell$ finite fields of the same degree.
Since all finite fields with the same number of elements are isomorphic, the
result follows.\qed
\end{proof}

\subsection{Transferring the setting of the attacks}
Considering the previous results, we describe a new setting for the attacks.
Let $f(x)$ be a monic polynomial of degree $N$ irreducible over $\mathbb{Z}[x]$,
$q$ a prime, and $\mathcal{G}$ a Gaussian distribution of mean $0$ and variance
$\sigma^2$. With these elements we have our $(f,q,\sigma)$-PLWE instance.
Now, let $g(x)$ be also a monic polynomial of degree $N$ irreducible over
$\mathbb{Z}[x]$ that has an $n$-ideal factor in $\mathbb{F}_q$ (with an
independent term of small multiplicative order) having the same
\emph{factorization structure} as $f(x)$. Then, we apply the following steps:
\begin{itemize}
\item
First, every $(f,q,\sigma)$-PLWE sample is transformed into a $(g,q,
\sigma)$-PLWE sample via the existing isomorphism between $R_q :=
\mathbb{F}_q[x]/(f(x))$ and $R_q^\prime := \mathbb{F}_q[x]/(g(x))$ by virtue of
Proposition~\ref{p.1}. This transformation should not affect the
distinguishability of the induced distribution in the sense that if the
distinguisher succeeds in the first setting, so will it succeed in the second,
which is critical to the applicability of the presented transformation.
\item
Then, we use the homomorphism derived from the evaluation at root $\alpha$ of
the $n$-ideal factor $x^n - a$ of $g(x)$ to migrate our samples into
$\mathbb{F}_{q^n}$.
\item
From that point on, we are entitled to apply any of the attacks referenced
in Section~\ref{S.2} in order to create a successful decisional attack, by resorting either
to the $\mathbb{F}_q$ attacks, or applying the trace over $\mathbb{F}_{q^n}$
for $n > 1$.
\end{itemize}
In short, we could transfer the attacks to a more favorable setting, namely, a
$(g,q,\sigma)$-PLWE instance in which the polynomial $g(x)$ has a suitable
$n$-ideal factor if it happens to exist.

It is noticeable that the existence of such polynomials $g(x)$ depends only
upon the parameters $(f,q)$, so that they could be precomputed (in case they
exist) and stored in advance.

\subsection{The fully-split setting and a generalization beyond it}
However nice and promising as it may sound, we will show in the sequel that the
previous results lead us again to a blind alley regarding the possibility of
achieving new, more powerful attacks.

To begin with, we consider the fully-split setting, namely, a setting where both
$f(x)$ and $g(x)$ decompose into
linear factors over $\mathbb{F}_q$ and the roots, say $(\alpha_1,\dotsc,
\alpha_N)$, $(\beta_1,\dotsc, \beta_N)$, respectively, are distinct. For this
case, we give in Appendix~\ref{a.1} a very detailed account about how to construct
an explicit family of isomorphisms, between $R_q$ and $R_q^\prime$ that,
moreover, happen to be unique.

The result laid out in the Appendix just referred to shows that the application of an
isomorphism, when composed with a subsequent evaluation at a certain root
$\beta$, yields precisely the evaluation of the original samples at the
corresponding root $\alpha$ (which will depend on the choice of $\beta$ and the
isomorphism). In other words, the attacks over the transferred setting reduce
to the original setting and so no advantage is gained with the transference.

Therefore, in spite of all nice (possibly of independent interest) results supplied
therein,
Section~\ref{4.3} and Section~\ref{4.4} give us sufficient evidence that, in the
fully-split setting, the use of isomorphisms of rings do not provide us with new ways
for mounting successful attacks against PLWE instances via the use of roots of
(fully-split) polynomials.

However this is not the end of the story: We can go a step further by analyzing how the combination of such morphisms look like. We need a number of lemmas, beginning with a simple result regarding the composition of homomorphisms:

\begin{lemma} \label{gen_comp_iso}
Let A, B and C be rings. If there exists a ring isomorphism $\psi\colon A
\rightarrow B$, then every ring homomorphism $\eta\colon A \rightarrow C$ arises as
the composition of $\psi$ and a certain ring homomorphism $\xi\colon B \rightarrow
C$.
\end{lemma}
\begin{proof}
Choose a homomorphism $\eta\colon A \rightarrow C$. Then, since $\psi\colon A
\rightarrow B$, $\psi^{-1}$ exists and is a ring isomorphism $B \rightarrow A$.
Then, $\eta \circ \psi^{-1}$ is a ring homomorphism $\xi: B \rightarrow C$.
Therefore, $\eta = \xi \circ \psi$.\qed
\end{proof}

The next result states that the only homomorphisms from polynomial rings are the evaluation homomorphisms:
\begin{lemma}\label{unique_eval_hom}
Let $q$ be a prime and $f(x)\in \mathbb{F}_q[x]$ a polynomial which factors
over $\mathbb{F}_q[x]$. Then, the only ring homomorphisms
\[
\textstyle\phi\colon\mathbb{F}_q[x]/(f(x)) \rightarrow \F_q^k
\]
are the evaluation homomorphisms $\ev_{\alpha_i}$, where $k$ denotes the degree of the irreducible polynomial in $\mathbb{F}_q[x]$ which has $\alpha_i$ as a root (i.e. the degree of the extension defined by $\alpha_i$ over $\mathbb{F}_q$).
\end{lemma}
\begin{proof}
Any ring homomorphism satisfies $\phi(1)=1$ hence $\phi|_{\mathbb{F}_q}=Id$.
Further, as $\phi$ is linear and compatible with multiplications, we have
that for any polynomial class
$g(\overline{x})=\sum_{i=0}^{N-1}g_i\overline{x}^i\in\mathbb{F}_q[x]/(f(x))$
$$
\phi(g(\overline{x}))=\sum_{i=0}^{N-1}g_i\phi(\overline{x})^i=g(\phi(\overline{x})),
$$
so that $\phi$ is determined by its image on the class $\overline{x}$. But since
$f(\overline{x})=0$, necessarily
$\phi(f(\overline{x}))=f(\phi(\overline{x}))=0$ and hence
$\phi(\overline{x})=\alpha_i$ for some $i$ among the roots of degree $k$ of
$f(x)$.\qed
\end{proof}

The next lemma states that the composition of any ring isomorphism
between polynomial rings and an evaluation homomorphism yields an evaluation
homomorphism.
\begin{lemma} \label{equality_evaluation}
Let $f(x)$ and $g(x)$ be two polynomials with the same factorization structure. Let $\psi$ be a ring isomorphism
\[
\textstyle\psi\colon \mathbb{F}_q[x]/(f(x))\rightarrow
\mathbb{F}_q[x]/(g(x))
\]
and $\phi:= \ev_{\beta_i}$ be an evaluation homomorphism
\[
\textstyle\phi\colon \mathbb{F}_q[x]/(\prod_{i=1}^N (x - \beta_i))\rightarrow \F_q^k
\]
for a root $\beta_i$ of $g(x)$ of degree $k$ over $\mathbb{F}_q$. Then, $\varphi := \ev_{\beta_i} \circ \psi$ is an evaluation homomorphism
\[
\textstyle\varphi\colon \mathbb{F}_q[x]/(\prod_{i=1}^N (x - \alpha_i)) \rightarrow \F_q^k
\]
over a certain root $\alpha_{j_{i}}$ of $f(x)$ of degree $k$ over $\mathbb{F}_q$.
\end{lemma}
\begin{proof}
It is a direct consequence of Lemmas \ref{gen_comp_iso} and \ref{unique_eval_hom}.\qed
\end{proof}

Together, we have the following corollary:
\begin{corollary}\label{c.1}
The equality
\[
\ev_{\alpha_{j_i}} = \ev_{\beta_i} \circ \psi
\]
holds for any ring isomorphism $\psi$ between polynomial rings and
$\alpha_{j_i}$ and $\beta_i$ roots of $f(x)$ and $g(x)$ with the same degree over $\mathbb{F}_q$,
respectively.
\end{corollary}

Observe that Corollary \ref{c.1} applies to any ring and not just to
$\F_q$. Moreover, the root-based attack on $k > 1$ uses a map that is the composition
of the evaluation at a certain root (which is now a homomorphism to the finite
extension $\F_{q^k}$) and the trace operator defined over such extension.
Thus,
\begin{equation}\label{eq.5}
(\Tr \circ \ev_{\beta_i}) \circ \psi = \Tr \circ (\ev_{\beta_i} \circ
\psi) = \Tr \circ \ev_{\alpha_{j_i}}
\end{equation}
and one falls back precisely to the original root-based attack.

Note that this result captures in full generality the nature of the proposed
attack extension: the application of the isomorphism to the samples in the
original (allegedly robust) setting plus the evaluation at roots of the
vulnerable PLWE instance. And the proof shows that this is precisely the same as
evaluating the original samples at the corresponding root of the original
instance.

In conclusion, this result allows us to generalize the analysis to instances
beyond fully-split ones, completely characterizing the behavior when $k > 1$,
and, in particular, corroborates the intuition that using isomorphisms in order
to transfer to another (potentially weaker) setting cannot provide new, more
powerful attacks over any type of PLWE instances, whether fully-split ones or
beyond.

\bibliographystyle{splncs04}
\bibliography{plwe}

\appendix

\section{The fully-split case}\label{a.1}

\subsection{The fully-split case: Construction of an explicit isomorphism}\label{4.3}
Suppose that both $f(x)$ and $g(x)$ decompose into linear factors over
$\mathbb{F}_q$ and the roots, $(\alpha_1,\dotsc,\alpha_N)$, $(\beta_1,\dotsc,
\beta_N)$, respectively, are distinct. Then, we are able to construct an
explicit family of isomorphisms between $R_q$ and $R_q^\prime$.

The particular family of isomorphisms constructed can be defined in three conceptual phases:
\begin{itemize}
\item
The isomorphism $\phi_1$, derived from the CRT,
\[
\phi_1\colon R_q \rightarrow \oplus_{i=1}^N\mathbb{F}_q.
\]
\item
The \emph{coordinate change} representation of elements in
$\oplus_{i=1}^N\mathbb{F}_q$ in terms of powers of the roots $\alpha_i$ of
$f(x)$ to powers of the roots $\beta_i$ of $g(x)$. Note that, since each
coordinate is an element of $\F_q$, it can be viewed as either an element of the
powers of $\alpha_i$ or $\beta_i$, and therefore this map is just the identity.
\item
The inverse isomorphism, $\phi_2^{-1}$ derived also from the CRT,
\[
\phi_2\colon R_q^\prime \rightarrow \oplus_{i=1}^N\mathbb{F}_q.
\]
\end{itemize}
Given $a(x)\in R_q$, then $\phi_1(a(x)) = (a(\alpha_1),\dotsc,a(\alpha_N))$
that can be naturally described by means of the Vandermonde matrix of the
roots of $f(x)$ in $\mathbb{F}_q$: Defining the vector $\mathbf{a}=(a_0,
\dotsc, a_{N-1})$, we have $\phi_1(a(x)) = V_f\cdot\mathbf{a}$, where
\[
V_f = \left[
\begin{array}{cccc}
1      & \alpha_1 & \cdots & \alpha_1^{N-1} \\
\vdots & \vdots   & \ddots & \vdots \\
1      & \alpha_N & \cdots & \alpha_N^{N-1}
\end{array}
\right],
\]
and, analogously, $V_g$ for the Vandermonde matrix of the roots of $g(x)$.

Observe that by hypothesis, all the roots are distinct, so that both $V_f$ and
$V_g$ are invertible over the field $\mathbb{F}_q$. Therefore the identity $I$
can be generated simply as $I = V_g\cdot V_g^{-1}$.

For the third step, we can consider now that $\phi_2$ can be described by means
of $V_g$, so that the complete isomorphism can be eventually written as
\[
\phi = \phi_2^{-1}\circ I\circ\phi_1 = \phi_2^{-1}\circ\phi_1 = V_g^{-1}\cdot V_f
\]

Before we continue, we provide two important properties of Vandermonde matrices.
\begin{proposition}\label{p.2}
Let $V_a$ and $V_b$ be two invertible Vandermonde matrices. Then, the matrix
$V_a^{-1}V_b$ has an identity vector as first column, namely, all coefficients
are zero except the first one
\end{proposition}
\begin{proof}
Let $V$ a Vandermonde matrix. By definition,
\[
\left[
\begin{array}{c}
1 \\
1 \\
\vdots \\
1
\end{array}
\right] = V\cdot
\left[
\begin{array}{c}
1 \\
0 \\
\vdots \\
0
\end{array}
\right]
\]
If $V$ is invertible, one has
\[
V^{-1}\cdot\left[
\begin{array}{c}
1 \\
1 \\
\vdots \\
1
\end{array}
\right] =
\left[
\begin{array}{c}
1 \\
0 \\
\vdots \\
0
\end{array}
\right]
\]
This means that the sum of each row of the inverse of a Vandermonde matrix is
$0$, except for the first row, which is $1$. Then, the right product with
another Vandermonde matrix (which has an all-$1$ first column), generates the
desired result.\qed
\end{proof}

Before proceeding let us recall the following
\begin{definition}\label{d.2}
The $j$-\emph{elementary symmetric polynomial} in $n$ variables is defined as
\[
E_j(X_1,\dotsc,X_n) = \sum_{1\leq a_1 < a_2 < \cdots < a_j \leq n}
X_{a_1}X_{a_2}\cdots X_{a_j},
\]
with $E_0(X_1,\dotsc,X_n) = 1$.
\end{definition}

We recall now a simple, yet useful lemma regarding elementary symmetric
polynomials:
\begin{lemma}\label{l.2}
The sum of the elementary symmetric polynomials over variables $X_1,\dotsc,
X_n$ weighted by $\lambda$ is equal to the evaluation at $\lambda$ of the
monic polynomial of roots $X_1, \dotsc, X_n$. In other words,
\[
\sum_{i=0}^n (-1)^{n-i}\cdot\lambda^i\cdot E_{n-i}(X_1,\dotsc,X_n) =
                       \prod_{i=1}^{n}(\lambda - X_i),
\]
with $E_0=1$ for any number of variables.
\end{lemma}
\begin{proof}
Expanding the product of monomials of the polynomial $P(t) = (t - X_1) \cdots
(t- X_n)$ and evaluating at $\lambda$ yields the desired result.\qed
\end{proof}

Based upon Definition~\ref{d.2}, we denote $E_{j,k}(X_1,\dotsc,X_n)$ as the
$j$-th elementary polynomial over the set of variables $\{X_1,\dotsc,X_n\}
\setminus \{X_k\}$, namely,
\[
E_{j,k}(X_1,\dotsc,X_n) := E_j(X_1,\dotsc,X_{k-1},X_{k+1},\dotsc,X_n).
\]

Let us define now the polynomial $G_k(\lambda;X_1,\dotsc,X_n) :=
\prod_{i=1,i\neq k}^{n}(\lambda - X_i)$. Then, we have the following
\begin{lemma}\label{l.3}
With the notations above, the following holds:
\[
\sum_{i=1}^n (-1)^{n-i}\cdot\lambda^{i-1}\cdot E_{n-i,k}(X_1,\dotsc,X_n) =
                                          G_k(\lambda; X_1,\dotsc,X_n),
\]
for any $k\in\{1,\dotsc,n\}$.

\end{lemma}
\begin{proof}
It suffices to expand the polynomial defined by $G_k$ and apply
Lemma~\ref{l.2}, remembering the definition of $E_{j,k}$.\qed
\end{proof}

\begin{proposition}\label{p.3}
The inverse of a Vandermonde matrix
\[
V = \left[
\begin{array}{cccc}
1      & \xi_1 & \cdots & \xi_1^{n-1} \\
\vdots & \vdots   & \ddots & \vdots \\
1      & \xi_n & \cdots & \xi_n^{n-1}
\end{array}
\right],
\]
is (when invertible) of the form
\begin{eqnarray*}
V^{-1}_{i,j} & = & (-1)^{i+j} \cdot \frac{E_{n-i, j}(\xi_1,\dotsc,\xi_n)}
                   {\prod_{l = 1}^{j-1}(\xi_j - \xi_l)
                   \cdot
                   \prod_{l = j+1}^{n} (\xi_l - \xi_j)} \\
             & = &  (-1)^{i+j} \cdot \frac{E_{n-i, j}(\xi_1,\dotsc,\xi_n)}
                   {\prod_{l < k}^{n} (\xi_k - \xi_l)}, \;
                   l = j \; \text{ or } \; k = j,
\end{eqnarray*}
where, as defined above, $E_{i, j}(\xi_1,\dotsc,\xi_n)$ represents the $i$-th
elementary symmetric polynomial over the set $\{\xi_1,\dotsc,\xi_n\} \setminus
\{\xi_j\}$.
\end{proposition}
\begin{proof}
See \cite{Rawashdeh:2019:ASM}, for example.\qed
\end{proof}

For the sake of brevity, we will define $Q_j(\xi_1,\dotsc,\xi_n)$ as
\begin{equation}\label{eq.1}
Q_j(\xi_1,\dotsc,\xi_n) := \prod_{l = 1}^{j-1}(\xi_j - \xi_l)
                           \cdot
                           \prod_{l = j+1}^{n} (\xi_l - \xi_j).
\end{equation}
Observe that for the Vandermonde matrix $V$ to be invertible all of the $\xi_1,
\dotsc,\xi_n$ must be distinct. Since by hypothesis, it is the case that all of
the $\beta_1,\dotsc, \beta_N$, the roots of $g(x)$, are indeed distinct, then it
follows that $V_g^{-1}$ does exist.

The idea of this attack is to choose a polynomial $g(x)$ with roots $\{\beta_1,
\dotsc, \beta_N\}$ such that the resulting samples, via the isomorphism between
$\F_q[x]/(f(x))$ and $\F_q[x]/(g(x))$, can be efficiently distinguished by resorting to the original
attacks of \cite{ELOS:2015:PWI} and \cite{ELOS:2016:RCN}. The difficulty lies
in the fact that the constructed isomorphisms will likely distort the samples
so as to render the attacks unpractical or unfeasible.

On the positive side, this attack has two configurable elements that will hold
the key to their applicability:
\begin{itemize}
\item
The root $\beta_i$ to use for the evaluation of the resulting samples via the
isomorphism will be chosen so as to maximize the success possibility. In other
words, since we can choose it as we like, this value will either be $\{0, 1,
-1\}$ or an element of, at most, order $3$.
\item
The remaining roots, which are not relevant to the distinguisher construction,
but will be critical in order to build an isomorphism that does not distort
the samples too much.
\end{itemize}

\subsubsection{Attack with evaluating root $\beta = 0$.}
As a starting point, we will consider whether the easiest possible configuration
can be applied, \emph{i.e.}, the use of $\beta_i = 0$ as the evaluating root.

It is a direct consequence of the last proposition that, if $\beta_i = 0$, then
the first row of the inverse matrix $V_g^{-1}$ has all its entries $0$ but for
the $i$-th column, which is $1$. And, this means that the first row of the
isomorphism matrix $V_g^{-1} \cdot V_f$ is, in turn,
\[
\left[
\begin{array}{ccccccc}
0 & \cdots & 0 & 1 & 0 & \cdots & 0
\end{array}
\right]
    \cdot
\left[
\begin{array}{cccc}
1      & \alpha_1 & \cdots & \alpha_1^{N-1} \\
\vdots & \vdots   & \ddots & \vdots \\
1      & \alpha_N & \cdots & \alpha_N^{N-1}
\end{array}
\right]
 =
\left[
\begin{array}{cccc}
1 & \alpha_i & \cdots & \alpha_i^{N-1}
\end{array}
\right].
\]

Now, if we are given a sample seen a as a vector, $\mathbf{a}=(a_0,\dotsc,
a_{N-1})$, and the transformed sample via the isomorphism, $\mathbf{b}=(b_0,
\dotsc,b_{N-1})$ and we evaluate the latter at the root $\beta = 0$, the only
non-zero term is
\[
b_0 =
\left[
\begin{array}{ccc}
a_0 & \cdots & a_{N-1}
\end{array}
\right]
    \cdot
\left[
\begin{array}{c}
1 \\
\alpha_i \\
\vdots \\
\alpha_i^{N-1}
\end{array}
\right] =
\sum _{j=0}^{N-1} a_j\cdot\alpha^j_i.
\]
But then we have reduced the problem to the original setting under evaluation
on $\alpha_i$. Since we assume the latter (or any other root of $f(x)$ for
that matter) to be unsuccessful, so the isomorphism provides no advantage for
the attack.

Therefore, we can claim that applying the isomorphism and evaluating over a
(potential) root $\beta=0$ will not yield a successful attack regardless of the
polynomial $g(x)$ chosen. This does not mean however that a polynomial with $0$
as root will always be unsuccessful (which is not necessarily true), it only
means that evaluation under $\beta = 0$ will be.

\subsubsection{Attack with evaluating root $\beta = 1$.}
We turn now to the next interesting root to be analyzed, namely, $\beta=1$. If
we denote the isomorphism matrix $M:=V_g^{-1}\cdot V_f$, as given by
\[
M =
\left[
\begin{array}{cccc}
1      & x_{1,2} & \cdots & x_{1,N} \\
0      & x_{2,2} & \cdots & x_{2,N} \\
\vdots & \vdots  & \ddots & \vdots   \\
0      & x_{N,2} & \cdots & x_{N,N} \\
\end{array}
\right],
\]
then we can write the transformation in matrix notation as
$\mathbf{b}=M\cdot\mathbf{a}$, namely,
\begin{eqnarray*}
b_0 & = & a_0 + \sum_{j=1}^{N-1} x_{1,j+1} \cdot a_j, \\
b_i & = & \sum_{j=1}^{N-1} x_{i,j+1} \cdot a_j, \quad\quad 2\leq i\leq N,\\
\end{eqnarray*}
so that the transformed sample in polynomial representation becomes
\begin{equation}\label{eq.3}
b(y) = a_0 + \sum_{i=1}^N y^{i-1}\sum_{j=1}^{N-1}x_{i,j+1} \cdot a_j.
\end{equation}

The evaluation of the latter at the root $\beta=1$ yields
\begin{equation}\label{eq.2}
b(1) = a_0 + \sum_{j=1}^{N-1}a_j\cdot\sum_{i=1}^N x_{i,j+1} = a_0 +
\sum_{j=1}^{N-1}a_j\cdot S_{j+1}(M),
\end{equation}
where $S_j(M) = \sum_{i=1}^N x_{i,j}$ is precisely the sum of all the entries
in the $j$-th column of the matrix $M$. Put more bluntly, each of the terms of
the original sample $\mathbf{a}$ (but for $a_0$) becomes ``multiplicatively
perturbed'' by the sum of the entries in the corresponding column of the
isomorphism matrix.

In this way, we are able to deduce the first condition in order to achieve a
practical attack when $\beta=1$ is selected as the evaluation root, namely,
select a polynomial $g(x)$ (or, rather, its roots) so that the sum of the
entries of each column of the isomorphism matrix yields (almost) the same
values. More precisely, keep
\[
\mid\{S_j(M)\colon 1\leq j\leq N\}\mid\;\leq\;T,
\]
where $T$ represents a threshold value not overpassing $3$. Moreover, not only
the number of distinct elements in the set above but also their value itself
will impact on the success probability of the attack (depending on the attack
chosen). The most desirable values are, obviously, $0$, $1$, or $-1$.

Moreover, we have the next
\begin{lemma}\label{l.4}
The following equation holds:
\[
Q_j(\beta_1,\dotsc,\beta_n) = (-1)^{N+j} G_j(\beta_j; \beta_1,\dotsc,\beta_N).
\]
\end{lemma}
\begin{proof}
The result follows from the definition of $Q_j$ (see Equation~\eqref{eq.1}) and
$G_j$, which are identical but for a number of inversions in $Q_j$, which is
precisely $N-j$. But recalling that $-1$ has order $2$ in $\mathbb{Z}^\ast$,
it is clear that $N-j\equiv N+j\bmod{2}$ and the result follows.\qed
\end{proof}

Equipped with these tools,
let us now provide a more in-depth description of the $x_{i,j}$ in the
isomorphism matrix $M$ and how the sums $S_j(M)$ look like, a critical step
towards analyzing the likelihood of solving the generated system. Each of the
entries of the isomorphism matrix can be represented as
\begin{equation}\label{eq.4}
x_{i,j} = \sum_{k=1}^N
(-1)^{i+k}\frac{\alpha_k^{j-1}}{Q_k(\beta_1,\dotsc,\beta_N)}E_{N-i,k}(\beta_1,
                                \dotsc,\beta_N),
\end{equation}
where $Q_k(\beta_1,\dotsc,\beta_N)$ follows its definition in
Equation~\ref{eq.1}.

In other words, the $[i,j]$ entry of the matrix is the sum of the $(N-i)$-th
elementary polynomial over all possible subsets of $N-1$ roots of $g(x)$, each
one weighted by the term $\alpha_{k}^{j-1}/Q_k$.

Thus, the sum $S_j(M)$ can be expressed as
\begin{eqnarray*}
S_j(M) & = & \sum_{i=1}^N\sum_{k=1}^N (-1)^{i+k}\frac{\alpha_k^{j-1}}{Q_k
(\beta_1,\dotsc,\beta_N)}E_{N-i,k}(\beta_1,\dotsc,\beta_N) \\
       & = & \sum_{k=1}^N \frac{\alpha_k^{j-1}}{Q_k (\beta_1,\dotsc,\beta_N)}
             \sum_{i=1}^N (-1)^{i+k}E_{N-i,k}(\beta_1,\dotsc,\beta_N).
\end{eqnarray*}
Informally, the $j$-th column sum, $S_j(M)$, is the sum of all symmetric
polynomials on each of the $N$ subsets, each sum being weighted by the term
$\alpha_{k}^{j-1}/Q_k$.

From the previous characterization, we extract two observations:
\begin{enumerate}
\item
Observe that for the attack to be applicable, we need most of the elements
$S_j(M)$ to be equal and it is clear from the computations above that all of
them are influenced by the common factor,
\[
SE_k := \sum_{i=1}^N (-1)^{i+k}E_{N-i,k}(\beta_1,\dotsc,\beta_N),
\]
for $k\in\{1,\dotsc,N\}$, which represent the sum of the symmetric polynomials
on each of the subsets. Thus, a clear initial path can be choosing the roots of
$g(x)$ is a way that these sums are all $0$.
\item
Since this common term appears on every $S_j(M)$ weighted with distinct values
that do depend on $j$, it would be most likely not possible for the elements
in the set $\{S_j(M)\}$ to be almost equal over a $0$-characteristic base field.
But, since we are working in a positive characteristic base field, it is still
possible.
\end{enumerate}

Remark that taking advantage of Lemma~\ref{l.3}, we can simplify the term
$SE_k$. Actually, evaluating at $\lambda=1$, we have
\[
G_k(1; \beta_1,\dotsc,\beta_N) =
\sum_{i=1}^N (-1)^{N-i} E_{N-i,k}(\beta_1,\dotsc,\beta_N),
\]
whence we deduce
\[
(-1)^{N+k}G_k(1; \beta_1,\dotsc,\beta_N) = \sum_{i=1}^N (-1)^{i+k}E_{N-i,k}
(\beta_1,\dotsc,\beta_N) = SE_k,
\]
recalling again that $-1$ has order $2$ in $\mathbb{Z}^\ast$ and $N+k+N-i\equiv
i+k\bmod{2}$. Accordingly
\[
S_j(M) = \sum_{k=1}^N (-1)^{N+k} \frac{\alpha_k^{j-1}}{Q_k (\beta_1,\dotsc,
         \beta_N)}G_k(1; \beta_1,\dotsc,\beta_N).
\]

Remember we are assuming that one of the $\beta$ roots has value $1$, say
$\beta_i=1$. Then, by its definition, $G_k(1;\beta_1,\dotsc,\beta_N) = 0$ for
all $k\neq i$. Hence
\[
S_j(M) = (-1)^{N+i} \frac{\alpha_i^{j-1}}{Q_i (\beta_1,\dotsc,
         \beta_N)}G_i(1; \beta_1,\dotsc,\beta_N),\quad\beta_i=1.
\]
But from Lemma~\ref{l.4} it follows directly that $Q_i(\beta_1,\dotsc, \beta_n)
= (-1)^{N+i} G_i(1; \beta_1,\dotsc,\beta_N)$, hence
\[
S_j(M) = \alpha_i^{j-1}.
\]
Remark that by hypothesis, all of the roots $\beta_1,\dotsc,\beta_N$ are
distinct, so $G_i(1;\beta_1,\dotsc,\beta_N)$ cannot be zero for any $i$.

Now we are in a position for rewriting Equation~\eqref{eq.2} as
\[
b(1) = \sum_{j=0}^{N-1}a_j\cdot \alpha_i^j,
\]
which, as discussed before, is not subject to any of the attacks, since the
original polynomial $f(x)$ is not. It is curious though how the exact same term
keeps appearing in all our approaches. This raises the question \emph{Will it
also appear on the general case?}

Anyway, this analysis shows that choosing $\beta = 1$ as the evaluation root
will not yield a satisfactory result, regardless of the polynomial $g(x)$.
As with $\beta = 0$, it does not mean, however, that $g(x)$ cannot have $\beta
= 1$ as root, it only means that it cannot be selected as the evaluating term,
although if selected, we already know how the sum of the columns will look like,
and we will probably feel uneasy should such value happens to occur.

\subsubsection{Attack with arbitrary evaluating root.}
We now choose, in general, any root $\xi$, hoping to avoid the
recurrent term of the sum of roots of $f(x)$. Once we migrate into this setting
(a root of order $> 1$), we need to truly consider the impact of the evaluation
root in the overall term. In the general case, following Equation~\eqref{eq.3},
we end up with evaluation terms of the form
\[
b(\xi) = a_0 + \sum_{i=1}^N \xi^{i-1}\sum_{j=1}^{N-1}x_{i,j+1} \cdot a_j
       = a_0 + \sum_{j=1}^{N-1}a_j\cdot S_{j+1,\xi}(M),
\]
which combined with Equation~\eqref{eq.4}, yields
\begin{eqnarray*}
S_{j,\xi}(M) = \sum_{i=1}^N \xi^{i-1}\cdot x_{i,j}
& = & \sum_{i=1}^N \xi^{i-1} \sum_{k=1}^N (-1)^{i+k} \frac{\alpha_k^{j-1}}
{Q_k(\beta_1,\dotsc,\beta_N)}E_{N-i,k}(\beta_1, \dotsc,\beta_N) \\
& = & \sum_{k=1}^N \frac{\alpha_k^{j-1}}{Q_k(\beta_1,\dotsc,\beta_N)}
\sum_{i=1}^N (-1)^{i+k} \xi^{i-1} E_{N-i,k}(\beta_1, \dotsc,\beta_N).
\end{eqnarray*}

Following again Lemma~\ref{l.3} and reasoning as above, we have
\[
(-1)^{N+k}G_k(\xi; \beta_1,\dotsc,\beta_N) =
\sum_{i=1}^N (-1)^{i+k} \xi^{i-1} E_{N-i,k}(\beta_1,\dotsc,\beta_N).
\]
The root $\xi$ is any arbitrary root of $g(x)$, say $\xi=\beta_\ell$. By
its definition, $G_k(\beta_\ell; \beta_1,\dotsc,\beta_N) = 0$ for all
$k\neq\ell$ so, keeping in mind Lemma~\ref{l.4}, we finally have
\[
S_{j,\beta_\ell}(M) =  (-1)^{N+\ell} \frac{\alpha_\ell^{j-1}}{Q_\ell
(\beta_1,\dotsc,\beta_N)}G_\ell(\beta_\ell; \beta_1,\dotsc,\beta_N) =
\alpha_\ell^{j-1}.
\]

Thus, for any arbitrary root of $g(x)$, $\beta_\ell$, we get the general
evaluated term
\[
b(\beta_\ell) = \sum_{j=0}^{N-1}a_j\cdot \alpha_\ell^j.
\]
It is interesting to note that the evaluation of $b(\beta_\ell)$ does not depend
on the value of the particular root selected but only on the corresponding root
(in this case $\alpha_\ell$) and the coefficients of the polynomial $f(x)$.

As discussed before, it is not subject to any of the attacks, since the original
polynomial $f(x)$ is not. Therefore, we have proven that fully-split polynomials
are secure, even in the face of constructing isomorphisms to other
\emph{vulnerable} settings.

This result would be of additional interest if one could show that every
isomorphism between the quotient polynomial fields of $\mathbb{F}_q$ generated
by $f(x)$ and $g(x)$ is of the form $V_{g}^{-1} \cdot V_f$, which would prove
that this settings are not vulnerable to further exploits of these attacks.
This is precisely what we attempt to do in the coming sections.

\subsection{The fully-split case: Uniqueness of isomorphism between the rings}\label{4.4}
In this section we prove that any isomorphism between our initial ring and the
target ring is given by the product of the inverse Vandermonde matrix of the
target polynomial and the Vandermonde matrix of the initial polynomial.
This settles once and for all the question that isomorphisms of polynomial
rings will not yield any successful advantage towards attacking polynomial
rings.

We state the main result as the following
\begin{theorem}\label{uniqueness_isomorphism}
Let $q$ be a prime and $f(x), g(x)$ be irreducible polynomials over $\Z[x]$
such that they factor completely over $\F_q[x]$ with $\{\alpha_i\}_{i=1}^N$
the set of roots of $f(x)$ in $\F_q[x]$, and $\{\beta_i\}_{i=1}^N$ the set of
roots of $g(x)$ in $\F_q[x]$. Then, any ring isomorphism
\[
\textstyle\psi\colon\mathbb{F}_q[x]/(\prod_{i=1}^N (x - \alpha_i)) \rightarrow
\mathbb{F}_q[x]/(\prod_{i=1}^N (x - \beta_i))
\]
is of the form $V_g^{-1} \cdot V_f$, where $V_f, V_g$ represent the
Vandermonde matrices of the roots of the (fully-split) polynomials $f(x)$ and
$g(x)$ over $\F_q[x]$.
\end{theorem}

In order to prove Theorem~\ref{uniqueness_isomorphism} we need the following
Lemma ensuring that every resulting homomorphism $\ev_{\alpha_{j_i}}$ is
distinct.

\begin{lemma} \label{distintc_roots}
For every two roots $\beta_1 \neq \beta_2$, the corresponding roots
$\alpha_{j_1}$ and $\alpha_{j_2}$ of the resulting evaluation homomorphisms
$\ev_{\alpha_{j_i}} := \ev_{\beta_{i}} \circ \psi$ are distinct.
\end{lemma}
\begin{proof}
Given $p(y) \in \mathbb{F}_q[y]/\prod_{i=1}^N (y - \beta_i)$, we have that
$\ev_{\beta_{1}}(p(y)) \neq \ev_{\beta_{2}}(p(y))$.
Since $\psi$ is an isomorphism of polynomial rings, $p(y)$ is the image of a
unique $r(x) \in \mathbb{F}_q[x]/\prod_{i=1}^N (x - \alpha_i)$. Therefore,
$\ev_{\beta_{1}}(\psi(r(x))) \neq \ev_{\beta_{2}}(\psi(r(x)))$.
And, by Lemma \ref{equality_evaluation}, $\ev_{\beta_{i}} \circ \psi =
\ev_{\alpha_{j_i}}$.

Therefore,
\[
\ev_{\alpha_{j_1}}(r(x)) \neq \ev_{\alpha_{j_2}}(r(x))
\]
and this means that $\alpha_{j_1} \neq
\alpha_{j_2}$.\qed
\end{proof}

With this auxiliary Lemma, we
can complete now the proof of Theorem~\ref{uniqueness_isomorphism}:

\begin{proof}[Theorem \ref{uniqueness_isomorphism}]
Now we want to explicitly compute how $\psi$ looks like. Given a certain $p(x)
\in \mathbb{F}_q[x]/\prod_{i=1}^N (x - \alpha_i)$, the image $\psi(p(x)) \in
\mathbb{F}_q[y]/\prod_{i=1}^N (y - \beta_i)$ can be expressed in terms of the
coefficients of the resulting polynomial.

Therefore, we view $\psi(p(x))$ as $(\psi_0(p(x)), \cdots, \psi_{N-1}(p(x)))$,
which is the coefficient representation of the polynomial $h(p(x)) :=
\sum_{k=0}^{N-1} \psi_k(p(x)) \cdot y^k$.

Now, by Corollary~\ref{c.1}, we have
\[
\ev_{\beta_i}(\psi(p(x))) = \ev_{\beta_i}(h(p(x))) =
\sum_{k=0}^{N-1} \psi_k(p(x)) \cdot \beta_i^k = \sum_{k=0}^{N-1} p_k
\alpha_{j_i}^k = \ev_{\alpha_{j_i}}(p(x))
\]
$\forall i \in \{1, \cdots N\}$. Thus, putting together this representation for
every $i$, we have:
\[
\left[
\begin{array}{cccc}
1      & \beta_1 & \cdots & \beta_1^{N-1} \\
\vdots & \vdots   & \ddots & \vdots \\
1      & \beta_N & \cdots & \beta_N^{N-1}
\end{array}
\right]
    \cdot
\left[
\begin{array}{c}
\psi_0(p(x)) \\
\psi_1(p(x))  \\
\vdots \\
\psi_{N-1}(p(x))
\end{array}
\right]
 =
\left[
\begin{array}{cccc}
1      & \alpha_{j_1} & \cdots & \alpha_{j_1}^{N-1} \\
\vdots & \vdots   & \ddots & \vdots \\
1      & \alpha_{j_N} & \cdots & \alpha_{j_N}^{N-1}
\end{array}
\right]
    \cdot
\left[
\begin{array}{c}
p_0 \\
p_1  \\
\vdots \\
p_{N-1}
\end{array}
\right]
\]
Since the roots $\{\beta_i\}_i$ are all distinct, we can invert the matrix and
have:
\[
\left[
\begin{array}{c}
\psi_0(p(x)) \\
\psi_1(p(x))  \\
\vdots \\
\psi_{N-1}(p(x))
\end{array}
\right]
 =
\left[
\begin{array}{cccc}
1      & \beta_1 & \cdots & \beta_1^{N-1} \\
\vdots & \vdots   & \ddots & \vdots \\
1      & \beta_N & \cdots & \beta_N^{N-1}
\end{array}
\right]^{-1}
    \cdot
\left[
\begin{array}{cccc}
1      & \alpha_{j_1} & \cdots & \alpha_{j_1}^{N-1} \\
\vdots & \vdots   & \ddots & \vdots \\
1      & \alpha_{j_N} & \cdots & \alpha_{j_N}^{N-1}
\end{array}
\right]
    \cdot
\left[
\begin{array}{c}
p_0 \\
p_1  \\
\vdots \\
p_{N-1}
\end{array}
\right]
\]
and, by Lemma \ref{distintc_roots}, we have that the roots $\{\alpha_{j_i}\}_i$
are all distinct. Therefore, the above matrices represent exactly the
Vandermonde matrices of the polynomials $f(x)$ and $g(x)$, arriving to the
desired result.\qed
\end{proof}

\subsection{A simple example}\label{4.5}
In the sequel we will present a simple example following the lines of
Proposition~\ref{p.1}.
The setting is $\mathbb{F}_q$ with $q=31$. We select the following $7$-degree
polynomial, which is irreducible over $\mathbb{Z}[x]$:
\[
f(x) = x^{7} + 25\,x^{6} + 9\,x^{5} + 11\,x^{4} + 19\,x^{3} + 8\,x^{2}
 + 21\,x + 29,
\]
but factors over $\mathbb{F}_q$ as
\begin{eqnarray*}
f_1(x) & = & x^{3} + 19 + 17\,x + 11\,x^{2}, \\
f_2(x) & = & x^{4} + 26 + 17\,x + 24\,x^{2} + 14\,x^{3}.
\end{eqnarray*}
with $f_1(x),f_2(x)\in\mathbb{F}_q[x]$. Then we have that
\[
R_q = \mathbb{F}_q[x]/\left(f_1(x)f_2(x)\right) \simeq
\mathbb{F}_q[x]/(f_1(x))\times \mathbb{F}_q[x]/(f_2(x)).
\]
Now we choose another polynomial, $g(x)$, featuring the same factorization
structure as $f(x)$, and irreducible over $\mathbb{Z}[x]$:
\[
g(x) = x^{7} + 29\,x^{6} + 8\,x^{5} + 9\,x^{4} + x^{3} + 22\,x^{2} + 23
\,x + 14,
\]
which factors over $\mathbb{F}_q$ as
\begin{eqnarray*}
g_1(x) & = & x^{3} - 5 \\
g_2(x) & = & x^{4} + 22 + 14\,x + 8\,x^{2} + 29\,x^{3},
\end{eqnarray*}
with $g_1(x),g_2(x)\in\mathbb{F}_q[x]$. Then we have
\[
R^\prime_q = \mathbb{F}_q[x]/\left(g_1(x)g_2(x)\right) \simeq
\mathbb{F}_q[x]/(g_1(x))\times \mathbb{F}_q[x]/(g_2(x))).
\]
Observe that $g_1(x)$ is a $3$-ideal factor and, moreover, the order of
$5$ in $\mathbb{F}_q$ is just $3$, as required for the successful
application of the suitable attacks.

Now, using the polynomial basis, we compute the isomorphism $\psi$ from $R_q$
to $R^\prime$. Since this is a toy example, the simplest algorithm is just
exhaustive search that eventually provides us with the following matrix over
the polynomial basis:
\[
\psi =
\left(
\begin{array}{rrrrrrr}
1 & 7 & 23 & 14 & 0 & 25 & 14 \\
0 & 0 & 28 & 18 & 12 & 0 & 20 \\
0 & 3 & 26 & 30 & 24 & 15 & 29 \\
0 & 26 & 12 & 17 & 15 & 0 & 17 \\
0 & 6 & 25 & 29 & 3 & 14 & 5 \\
0 & 17 & 23 & 17 & 1 & 17 & 26 \\
0 & 20 & 9 & 28 & 22 & 13 & 25
\end{array}
\right).
\]
Besides, it is very easy to compute its inverse in $GL(\mathbb{F}_q)$, which turns
out to be
\[
\psi^{-1} =
\left(
\begin{array}{rrrrrrr}
1 & 2 & 13 & 15 & 1 & 10 & 19 \\
0 & 1 & 20 & 5 & 5 & 24 & 26 \\
0 & 19 & 10 & 25 & 17 & 21 & 19 \\
0 & 19 & 7 & 21 & 28 & 15 & 28 \\
0 & 5 & 10 & 11 & 19 & 19 & 29 \\
0 & 0 & 25 & 12 & 9 & 4 & 4 \\
0 & 20 & 14 & 17 & 14 & 17 & 30
\end{array}
\right).
\]
Thus, if we take a random element in $R_q$, such as
\[
fs(x) = 28\,x^{6} + 19\,x^{5} + 10\,x^{4} + 2\,x^{3} + 15
\,x^{2} + 14\,x + 18,
\]
it is easy to compute $gs(x) = \psi(fs(x))\in R_q^\prime$, which yields
\[
gs(x) = 26\,x^{6} + 4\,x^{5} + 23\,x^{4} + 26\,x^{3} + 20
\,x + 23.
\]

Now let $\alpha$ be a root of $f_1(x)$, which obviously is also a root of
$f(x)$. We can consider the evaluation mapping
\[
\mathrm{ev}_\alpha\colon R_q\rightarrow\mathbb{F}_{q^3}.
\]
Applying $\mathrm{ev}_\alpha$ on $fs(x)$, we get
\[
\mathrm{ev}_{\alpha}(fs(x)) = fs(\alpha) = \alpha ^{2} + 18\,\alpha  + 3,
\]
which is an element living in $\mathbb{F}_{q^3}$, and represented in a
polynomial basis. In the same line, let $\beta$ be a root of $g_1(x)$, also a root of
$g(x)$, and consider the corresponding evaluation mapping $\mathrm{ev}_\beta$. Now if we
apply such mapping to $gs(x)$, we get
\[
\mathrm{ev}_{\beta}(gs(x)) = gs(\beta) = 20\,\beta ^{2} + 11\,\beta  + 28,
\]
again represented in a polynomial basis.

Now we are interested in checking Equation~\eqref{eq.5} in Corollary~\ref{c.1},
so we face the task of computing the traces of $fs(\alpha)$ and $gs(\beta)$ over
$\mathbb{F}_{q^3}$. The
simplest way is to compute the matrices associated to the endomorphisms
``multiply by $fs(\alpha)$'' and ``multiply by $gs(\beta)$'' and obtain their
traces, keeping in all the computations $\mathbb{F}_q$ as the base field.

By so doing, we get that first matrix is
\[
\mathrm{End}(fs(\alpha)) = 
\left(
\begin{array}{rrr}
3 & 12 & 22 \\
18 & 17 & 17 \\
1 & 7 & 2
\end{array}
\right),
\]
whereas the second is
\[
\mathrm{End}(gs(\beta)) = 
\left(
\begin{array}{rrr}
28 & 7 & 24 \\
11 & 28 & 7 \\
20 & 11 & 28
\end{array}
\right),
\]
and a simple computation shows that
\[
\Tr(fs(\alpha)) = \mathrm{tr}\left(\mathrm{End}(fs(\alpha))\right) =
\mathrm{tr}\left(\mathrm{End}(gs(\beta))\right) = \Tr(gs(\beta)) = 22.
\]

Remembering that $gs(x) = \psi(fs(x))$, we check for this particular case that
\[
\Tr\circ\;\mathrm{ev}_\alpha = \Tr\circ\;\mathrm{ev}_\beta\circ\psi,
\]
as desired.
\end{document}